\documentclass{article}
\usepackage{graphicx,epsfig,color}
\usepackage{graphics,latexsym}
\usepackage{graphicx}

\newcounter{savenumi}

\newtheorem{theoremfoo}{Theorem}
\newenvironment{theorem}{\pagebreak[1]\begin{theoremfoo}}{\end{theoremfoo}}

\newtheorem{propositionfoo}[theoremfoo]{Proposition}

\newtheorem{lemmafoo}[theoremfoo]{Lemma}
\newenvironment{lemma}{\pagebreak[1]\begin{lemmafoo}}{\end{lemmafoo}}
\newtheorem{conjecturefoo}[theoremfoo]{Conjecture}

\newtheorem{corollaryfoo}[theoremfoo]{Corollary}

\newtheorem{exercisefoo}{Exercise}

\newtheorem{openfoo}[theoremfoo]{Question}

\newtheorem{nttn}[theoremfoo]{Notation}

\newtheorem{dfntn}[theoremfoo]{Definition}
\newenvironment{definition}{\pagebreak[1]\begin{dfntn}\rm}{\end{dfntn}}

\newenvironment{proof}
    {\pagebreak[1]{\narrower\noindent {\bf Proof:\quad\nopagebreak}}}{\QED}

\newcommand{\ceiling}[1]{\left\lceil#1\right\rceil}
\def\nre.{$n$\/-r.e.}

\newtheorem{factfoo}[theoremfoo]{Fact}

\newcommand{\squeeze}{
\textwidth 6in
\textheight 8.8in
\oddsidemargin 0.2in
\topmargin -0.4in
}

\newtheorem{propertyfoo}[theoremfoo]{Property}

\makeatletter

\def\@makechapterhead#1{ \vspace*{50pt} { \parindent 0pt \raggedright 
 \ifnum \c@secnumdepth >\m@ne \huge\bf \@chapapp{} \thechapter. \par 
 \vskip 20pt \fi \Huge \bf #1\par 
 \nobreak \vskip 40pt } }

\def\@sect#1#2#3#4#5#6[#7]#8{\ifnum #2>\c@secnumdepth
     \def\@svsec{}\else 
     \refstepcounter{#1}\edef\@svsec{\csname the#1\endcsname.\hskip 1em }\fi
     \@tempskipa #5\relax
      \ifdim \@tempskipa>\z@ 
        \begingroup #6\relax
          \@hangfrom{\hskip #3\relax\@svsec}{\interlinepenalty \@M #8\par}
        \endgroup
       \csname #1mark\endcsname{#7}\addcontentsline
         {toc}{#1}{\ifnum #2>\c@secnumdepth \else
                      \protect\numberline{\csname the#1\endcsname}\fi
                    #7}\else
        \def\@svsechd{#6\hskip #3\@svsec #8\csname #1mark\endcsname
                      {#7}\addcontentsline
                           {toc}{#1}{\ifnum #2>\c@secnumdepth \else
                             \protect\numberline{\csname the#1\endcsname}\fi
                       #7}}\fi
     \@xsect{#5}}

\def\@begintheorem#1#2{\it \trivlist \item[\hskip \labelsep{\bf #1\ #2.}]}

\def\@opargbegintheorem#1#2#3{\it \trivlist
      \item[\hskip \labelsep{\bf #1\ #2\ (#3).}]}

\makeatother

\newif\ifshortconferences
\shortconferencesfalse
\newif\ifmediumconferences
\mediumconferencesfalse

\def\ending#1{{\count1=#1\relax
\count2=\count1
\divide\count2 by 100
\multiply\count2 by 100
\advance\count1 by -\count2
\ifnum\count1=11
th%
\else \ifnum\count1=12
th%
\else \ifnum\count1=13
th%
\else 
\count2=\count1
\divide\count1 by 10
\multiply\count1 by 10
\advance\count2 by -\count1
\ifnum\count2=1
st%
\else \ifnum\count2=2
nd%
\else \ifnum\count2=3
rd%
\else th%
\fi\fi\fi\fi\fi\fi
}}

\def\Proceedingsofthe{\ifshortconferences Proc.\else\ifmediumconferences Proc.\else Proceedings of the\fi\fi}

\newcounter{confnum}

\def\conf#1#2{%
\setcounter{confnum}{#2}%
\addtocounter{confnum}{-\csname #1zero\endcsname}%
\ifnum\value{confnum}=1%
\expandafter\ifx\csname #1One\endcsname\relax%
\Proceedingsofthe\ \arabic{confnum}\ending{\value{confnum}}\ \csname #1name\endcsname%
\else \csname #1One\endcsname\fi%
\else%
\Proceedingsofthe\
\arabic{confnum}\ending{\value{confnum}}\ \csname #1name\endcsname\fi}

\def\qsym{\vrule width0.7ex height0.9em depth0ex}

\newif\ifqed\qedtrue

\def\noqed{\global\qedfalse}

\def\qed{\ifqed{\penalty1000\unskip\nobreak\hfil\penalty50
\hskip2em\hbox{}\nobreak\hfil\qsym
\parfillskip=0pt \finalhyphendemerits=0\par\medskip}\fi\global\qedtrue}

\makeatletter
\def\eqnqed{\noqed
	\def\@tempa{equation}
	\ifx\@tempa\@currenvir\def\@eqnnum{\qsym}%
	\addtocounter{equation}{-1}\else%
    \def\@@eqncr{\let\@tempa\relax
    \ifcase\@eqcnt \def\@tempa{& & &}\or \def\@tempa{& &}%
      \else \def\@tempa{&}\fi
     \@tempa {\def\@eqnnum{{\qsym}}\@eqnnum}
     \global\@eqnswtrue\global\@eqcnt\z@\cr}\fi}

\def\eqnlabel#1#2{\if@filesw {\let\thepage\relax%
   \def\protect{\noexpand\noexpand\noexpand}%
   \edef\@tempa{\write\@auxout{\string
      \newlabel{#2}{{{#1}}{\thepage}}}}%
   \expandafter}\@tempa%
   \if@nobreak \ifvmode\nobreak\fi\fi\fi%
	\def\@tempa{equation}
	\ifx\@tempa\@currenvir\def\theequation{{#1}}%
	\addtocounter{equation}{-1}\else%
    \def\@@eqncr{\let\@tempa\relax
    \ifcase\@eqcnt \def\@tempa{& & &}\or \def\@tempa{& &}%
      \else \def\@tempa{&}\fi
     \@tempa {\def\@eqnnum{{#1}}\@eqnnum}
     \global\@eqnswtrue\global\@eqcnt\z@\cr}\fi}

\makeatother

\def\QED{\qed}

\makeatother

\squeeze

\begin{document}


\title{A Model for Donation Verification}

\author{Bin Fu$^1$, Fengjuan Zhu$^2$, and John Abraham$^1$
 \\ \\
$^1$Department of Computer Science\\
 University of Texas Rio Grande Valley\\
bin.fu@utrgv.edu, john.abraham@utrgv.edu \\
\\
$^2$Department of Law\\
Shaoxing University, Shaoxin, P. C. China\\
zhufj@usx.edu.cn
} \maketitle

\begin{abstract}
In this paper, we introduce a model for donation verification. A
randomized algorithm is developed to check if the money claimed
being received by the collector is $(1-\epsilon)$-approximation to
the total amount money contributed by the donors. We also derive
some negative results that show it is impossible to verify the
donations under some circumstances.
\end{abstract}


\section{Introduction}




Worldwide billions of dollars are donated for charities. For
example, United States alone gave over 335 billion dollars for
philanthropy in 2013. When this much money is involved there would
also be fraudsters who take advantage of one's generosity.
Recognizing fraudulent practices, US Federal Trade Commission has
given a number of things to check before giving to charity. The
efficiency of a charitable organization is currently determined by
the percentage of fund actually end up being used for intended
purpose. CharityWatch [http://www.charitywatch.org/criteria.html]
concludes 60\% or greater spent on charitable programs and the
remaining spent on overhead is acceptable. However, currently no
algorithms are available to detects errors in reporting of monies
donated. Donors merely trust the data provided by the charitable
organizations or charts published by organizations such as Charity
Navigator [http://www.charitynavigator.org]. Some research regarding
charity donations and their management have been conducted in the
academic
community~\cite{Wojciechowski09,CooterBroughman05,RanganathanHenley07,OlsenPracejusBrown03,Bennett09}.
We have not seen any existing research about how donors check the
amount of money received by the collector. It is essential to
develop some algorithm that the donors and charitable organizations
use to trust each other.

With the development of charity donations in the modern society, it
becomes a more and more important social problem about charity
donation system. In addition to establishing related laws, it is
also essential to build up efficient auditing systems about charity
donations, and apply big data technology to manage them. The
progression in this direction will bring  efficient and accurate
methods for charity donations, which will improve our social
reliability.

In this paper we develop a method that would allow us to verify
monies received by charitable organizations.  It would be difficult
for every donor to verify each philanthropic organization.  Our
method is based on a randomized model thereby reducing the number of
verifications.   Using our algorithm, even if only a small
percentage of the donors participate in the verification process,
incorrect data given by the philanthropic organizations (cheating)
can be detected.  With just few steps a donor can verify if the
money is used for intended purpose with a high degree of
probability.

\vskip 10pt

\begin{figure}
{ \vskip 100pt

  {\begin{picture}(10.0,10.0)



    \put(185.0, 100.0){\begin{picture}(0.0,0.0)
                           \line(1,0){60.0}
                         \end{picture}
                        }

    \put(185.0, 130.0){\begin{picture}(0.0,0.0)
                           \line(1,0){60.0}
                         \end{picture}
                        }

    \put(185.0, 100.0){\begin{picture}(0.0,0.0)
                           \line(0,1){30.0}
                         \end{picture}
                        }

    \put(245.0, 130.0){\begin{picture}(0.0,0.0)
                           \line(0,-1){30.0}
                         \end{picture}
                        }

    \put(210.0, 120.0){a}
    \put(195.0, 105.0){\small D(a)=100}


    \put(95.0, 50.0){\begin{picture}(0.0,0.0)
                           \line(1,0){60.0}
                         \end{picture}
                        }

    \put(95.0, 80.0){\begin{picture}(0.0,0.0)
                           \line(1,0){60.0}
                         \end{picture}
                        }

    \put(95.0, 50.0){\begin{picture}(0.0,0.0)
                           \line(0,1){30.0}
                         \end{picture}
                        }

    \put(155.0, 80.0){\begin{picture}(0.0,0.0)
                           \line(0,-1){30.0}
                         \end{picture}
                        }

    \put(120.0, 70.0){b}
    \put(105.0, 55.0){\small D(b)=6}

     \put(150.0, 80.0){\begin{picture}(0.0,0.0)
                          \vector(2,1){40.0}
                        \end{picture}
                       }


    \put(275.0, 50.0){\begin{picture}(0.0,0.0)
                           \line(1,0){60.0}
                         \end{picture}
                        }

    \put(275.0, 80.0){\begin{picture}(0.0,0.0)
                           \line(1,0){60.0}
                         \end{picture}
                        }

    \put(275.0, 50.0){\begin{picture}(0.0,0.0)
                           \line(0,1){30.0}
                         \end{picture}
                        }

    \put(335.0, 80.0){\begin{picture}(0.0,0.0)
                           \line(0,-1){30.0}
                         \end{picture}
                        }

    \put(300.0, 70.0){c}
    \put(285.0, 55.0){\small D(c)=94}

    \put(278.0, 80.0){\begin{picture}(0.0,0.0)
                          \vector(-2,1){40.0}
                        \end{picture}
                       }


    \put(35.0, 00.0){\begin{picture}(0.0,0.0)
                           \line(1,0){60.0}
                         \end{picture}
                        }

    \put(35.0, 30.0){\begin{picture}(0.0,0.0)
                           \line(1,0){60.0}
                         \end{picture}
                        }

    \put(35.0, 0.0){\begin{picture}(0.0,0.0)
                           \line(0,1){30.0}
                         \end{picture}
                        }

    \put(95.0, 30.0){\begin{picture}(0.0,0.0)
                           \line(0,-1){30.0}
                         \end{picture}
                        }

    \put(60.0, 20.0){d}
    \put(45.0, 5.0){\small D(d)=1}

     \put(70.0, 30.0){\begin{picture}(0.0,0.0)
                          \vector(2,1){40.0}
                        \end{picture}
                       }


    \put(135.0, 00.0){\begin{picture}(0.0,0.0)
                           \line(1,0){60.0}
                         \end{picture}
                        }

    \put(135.0, 30.0){\begin{picture}(0.0,0.0)
                           \line(1,0){60.0}
                         \end{picture}
                        }

    \put(135.0, 00.0){\begin{picture}(0.0,0.0)
                           \line(0,1){30.0}
                         \end{picture}
                        }

    \put(195.0, 30.0){\begin{picture}(0.0,0.0)
                           \line(0,-1){30.0}
                         \end{picture}
                        }

    \put(160.0, 20.0){e}
    \put(145.0, 5.0){\small D(e)=5}

     \put(175.0, 30.0){\begin{picture}(0.0,0.0)
                          \vector(-2,1){40.0}
                        \end{picture}
                       }


    \put(235.0, 00.0){\begin{picture}(0.0,0.0)
                           \line(1,0){60.0}
                         \end{picture}
                        }

    \put(235.0, 30.0){\begin{picture}(0.0,0.0)
                           \line(1,0){60.0}
                         \end{picture}
                        }

    \put(235.0, 0.0){\begin{picture}(0.0,0.0)
                           \line(0,1){30.0}
                         \end{picture}
                        }

    \put(295.0, 30.0){\begin{picture}(0.0,0.0)
                           \line(0,-1){30.0}
                         \end{picture}
                        }

    \put(260.0, 20.0){f}
    \put(245.0, 5.0){\small D(f)=10}

     \put(255.0, 30.0){\begin{picture}(0.0,0.0)
                          \vector(2,1){40.0}
                        \end{picture}
                       }


    \put(335.0, 00.0){\begin{picture}(0.0,0.0)
                           \line(1,0){60.0}
                         \end{picture}
                        }

    \put(335.0, 30.0){\begin{picture}(0.0,0.0)
                           \line(1,0){60.0}
                         \end{picture}
                        }

    \put(335.0, 00.0){\begin{picture}(0.0,0.0)
                           \line(0,1){30.0}
                         \end{picture}
                        }

    \put(395.0, 30.0){\begin{picture}(0.0,0.0)
                           \line(0,-1){30.0}
                         \end{picture}
                        }

    \put(360.0, 20.0){g}
    \put(345.0, 5.0){\small D(g)=84}

     \put(355.0, 30.0){\begin{picture}(0.0,0.0)
                          \vector(-2,1){40.0}
                        \end{picture}
                       }

   \end{picture}
  }\caption{Donation Tree with D(.) Values}\label{fig1}
}
\end{figure}
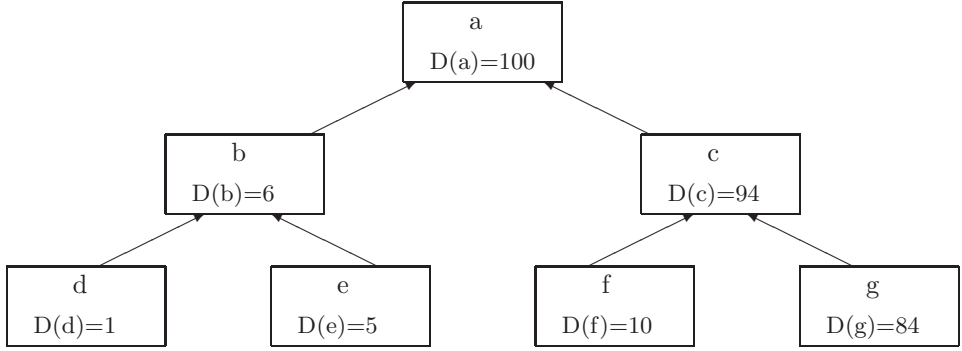

\section{Models}

Assume that there are $n$ people who donated money. Person $i$
donates $s_i$. In this model, we assume that each person checks his
donation with probability $1-e^{-\lambda s}$ if he donated $s$
amount money, where $\lambda$ is fixed. This model means that a
person will have larger chance to check his donation if he
contributes more money. We define a donation tree.


\begin{definition}We define a {\it donation tree}.
 For each leaf $L$ in a donation tree,
its donation value is defined to be $D(L)$. For node $N$ in a tree
$T$, define $D(N)$ to be the sum of values $D(.)$ in its leaves of
the subtree with root at $N$.  For each node $D$, function $V(N)$ is
the amount money that the collector claims from the donors at the
leaves of the subtree with root at $N$. An {\it error path} from a
leaf to the root has a node $N$ with $V(N)<V(N_1)+\cdots+V(N_k)$,
where $N_1,\cdots, N_k$ are the children of $N$ in $T$, and $V(N)$
be the vale saved in node $N$.
\end{definition}

A donation tree without cheating should be the case $D(N)=V(N)$ for
all nodes in the tree. We have the donation tree without cheating at
Figure~\ref{fig1}.


Let $k$ be an integer at least $2$. A $k$-donation tree is a
donation tree such that each internal node has at most $k$ children.
The money donated from one donor is at a leaf. For every node saves
$N$, the total money $D(N)$ of leave below it, and a value $V(N)\le
D(N)$ to represents the amount of money the collector claiming to
have received from the leaves. In the case $V(N)<D(N)$, it is
considered a cheating from the collector.



\section{Algorithm and Its Analysis}

In this section, we develop an algorithm for this problem.


\begin{lemma}\label{lemma-basic}For each integer $k\ge 2$,
A $k$-donation tree can be built in $O(n)$ time offline with depth
$O(\log n)$. It also supports an $O(\log n)$ time for both insertion
and deletion.
\end{lemma}

\begin{proof}
A divide and conquer method can be used to build a donation tree of
depth $O(\log n)$ with $O(n)$ time offline (the input of donations
from $n$ people are given). If it based on the structure of B+-tree,
then it can support $O(\log n)$ time for both insertion and
deletion.
\end{proof}

\vskip 10pt

{\bf Protocol}

 Collector:

\qquad Generate a file $F$ that a donation tree.

\qquad Publish the file $F$.

\vskip 10pt

donor $m_i$:

\qquad Check if his donation $(D(m_i)=V(m_i))$.

\qquad For each node $N$ on the path from $m_i$ to the tree root
$m_0$, check if $V(N)=V(N_1)+\cdots+V(N_k)$, where $N_1,\cdots,N_k$
are all children of $N$.

\qquad Report error path if at least one of the two checks fails.

{\bf End of Protocol}

\begin{lemma}\label{count-lemma}
Assume that the root $N$ of the donation tree $T$ has a value
$V(N)<D(N)$. There are leaves $m_1,\cdots, m_k$ in the tree that all
have error paths to root, and $D(m_1)+\cdots+D(m_k)\ge D(N)-V(N)$.
\end{lemma}

\begin{proof}
We prove it by induction. It is trivial when the depth is $0$.
Assume that the statement is true for the depth at most $d$.
Consider the depth $d+1$. Let $N$ be a node of depth $d+1$ and has
children $N_1,\cdots, N_k$.

Case 1.  $V(N)<V(N_1)+\cdots+V(N_k)$, then every leaf has a error
path.

Case 2. $V(N)\ge V(N_1)+\cdots+V(N_k)$.  Let  $N_{i_1},\cdots,
N_{i_t}$ be all of the nodes of $N_1,\cdots, N_k$ such that
$V(N_{i_s})<D(N_{i_s})$. We have
\begin{eqnarray}
\sum_{j=1}^t (D(N_{i_j})-V(N_{i_j}))&\ge& \sum_{a=1}^k
(D(N_{a})-V(N_{a}))\\
&\ge& D(N)-V(N).
\end{eqnarray}
  We note that for each $a\in \{1,2,\cdots,
k\}-\{i_1,\cdots, i_t\}$, $D(N_{a})-V(N_{a})\le 0$.

By induction hypothesis, for each $i_j$, there are leave
$l_{i_j,1},\cdots, l_{i_j, u}$  under the subtree with root at
$N_{i_j}$ such that $D(l_{i_j,1})+\cdots+D(l_{i_j, u})\ge
D(N_{i_j})-V(N_{i_j})$. Let $H_j=\{l_{i_j,1},\cdots, l_{i_j, u}\}$
for $j=1,2,\cdots, t$. Let $H$ be the set of all leave $m_i\in
H_1\cup\cdots H_t$, we have $\sum_{m_i\in H}D(m_i)\ge D(N)-V(N)$.

\end{proof}

\vskip 30pt

\begin{figure}
{ \vskip 100pt

  {\begin{picture}(10.0,10.0)



    \put(185.0, 100.0){\begin{picture}(0.0,0.0)
                           \line(1,0){60.0}
                         \end{picture}
                        }

    \put(185.0, 130.0){\begin{picture}(0.0,0.0)
                           \line(1,0){60.0}
                         \end{picture}
                        }

    \put(185.0, 100.0){\begin{picture}(0.0,0.0)
                           \line(0,1){30.0}
                         \end{picture}
                        }

    \put(245.0, 130.0){\begin{picture}(0.0,0.0)
                           \line(0,-1){30.0}
                         \end{picture}
                        }

    \put(210.0, 120.0){a}
    \put(195.0, 105.0){\small V(a)=100}


    \put(95.0, 50.0){\begin{picture}(0.0,0.0)
                           \line(1,0){60.0}
                         \end{picture}
                        }

    \put(95.0, 80.0){\begin{picture}(0.0,0.0)
                           \line(1,0){60.0}
                         \end{picture}
                        }

    \put(95.0, 50.0){\begin{picture}(0.0,0.0)
                           \line(0,1){30.0}
                         \end{picture}
                        }

    \put(155.0, 80.0){\begin{picture}(0.0,0.0)
                           \line(0,-1){30.0}
                         \end{picture}
                        }

    \put(120.0, 70.0){b}
    \put(105.0, 55.0){\small V(b)=6}

     \put(150.0, 80.0){\begin{picture}(0.0,0.0)
                          \vector(2,1){40.0}
                        \end{picture}
                       }


    \put(275.0, 50.0){\begin{picture}(0.0,0.0)
                           \line(1,0){60.0}
                         \end{picture}
                        }

    \put(275.0, 80.0){\begin{picture}(0.0,0.0)
                           \line(1,0){60.0}
                         \end{picture}
                        }

    \put(275.0, 50.0){\begin{picture}(0.0,0.0)
                           \line(0,1){30.0}
                         \end{picture}
                        }

    \put(335.0, 80.0){\begin{picture}(0.0,0.0)
                           \line(0,-1){30.0}
                         \end{picture}
                        }

    \put(300.0, 70.0){c}
    \put(285.0, 55.0){\small V(c)=94}

    \put(278.0, 80.0){\begin{picture}(0.0,0.0)
                          \vector(-2,1){40.0}
                        \end{picture}
                       }


    \put(35.0, 00.0){\begin{picture}(0.0,0.0)
                           \line(1,0){60.0}
                         \end{picture}
                        }

    \put(35.0, 30.0){\begin{picture}(0.0,0.0)
                           \line(1,0){60.0}
                         \end{picture}
                        }

    \put(35.0, 0.0){\begin{picture}(0.0,0.0)
                           \line(0,1){30.0}
                         \end{picture}
                        }

    \put(95.0, 30.0){\begin{picture}(0.0,0.0)
                           \line(0,-1){30.0}
                         \end{picture}
                        }

    \put(60.0, 20.0){d}
    \put(45.0, 5.0){\small V(d)=1}

     \put(70.0, 30.0){\begin{picture}(0.0,0.0)
                          \vector(2,1){40.0}
                        \end{picture}
                       }


    \put(135.0, 00.0){\begin{picture}(0.0,0.0)
                           \line(1,0){60.0}
                         \end{picture}
                        }

    \put(135.0, 30.0){\begin{picture}(0.0,0.0)
                           \line(1,0){60.0}
                         \end{picture}
                        }

    \put(135.0, 00.0){\begin{picture}(0.0,0.0)
                           \line(0,1){30.0}
                         \end{picture}
                        }

    \put(195.0, 30.0){\begin{picture}(0.0,0.0)
                           \line(0,-1){30.0}
                         \end{picture}
                        }

    \put(160.0, 20.0){e}
    \put(145.0, 5.0){\small V(e)=5}

     \put(175.0, 30.0){\begin{picture}(0.0,0.0)
                          \vector(-2,1){40.0}
                        \end{picture}
                       }


    \put(235.0, 00.0){\begin{picture}(0.0,0.0)
                           \line(1,0){60.0}
                         \end{picture}
                        }

    \put(235.0, 30.0){\begin{picture}(0.0,0.0)
                           \line(1,0){60.0}
                         \end{picture}
                        }

    \put(235.0, 0.0){\begin{picture}(0.0,0.0)
                           \line(0,1){30.0}
                         \end{picture}
                        }

    \put(295.0, 30.0){\begin{picture}(0.0,0.0)
                           \line(0,-1){30.0}
                         \end{picture}
                        }

    \put(260.0, 20.0){f}
    \put(245.0, 5.0){\small V(f)=10}

     \put(255.0, 30.0){\begin{picture}(0.0,0.0)
                          \vector(2,1){40.0}
                        \end{picture}
                       }


    \put(335.0, 00.0){\begin{picture}(0.0,0.0)
                           \line(1,0){60.0}
                         \end{picture}
                        }

    \put(335.0, 30.0){\begin{picture}(0.0,0.0)
                           \line(1,0){60.0}
                         \end{picture}
                        }

    \put(335.0, 00.0){\begin{picture}(0.0,0.0)
                           \line(0,1){30.0}
                         \end{picture}
                        }

    \put(395.0, 30.0){\begin{picture}(0.0,0.0)
                           \line(0,-1){30.0}
                         \end{picture}
                        }

    \put(360.0, 20.0){g}
    \put(345.0, 5.0){\small V(g)=84}

     \put(355.0, 30.0){\begin{picture}(0.0,0.0)
                          \vector(-2,1){40.0}
                        \end{picture}
                       }

   \end{picture}
  }\caption{Donation Tree without Cheating}\label{fig2}
}
\end{figure}
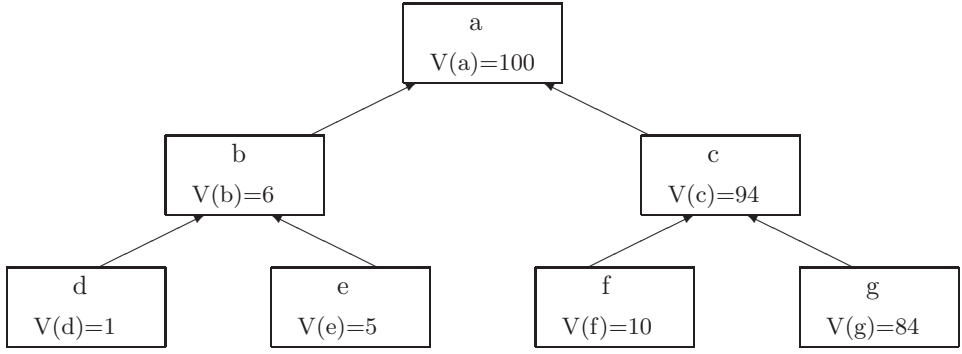

\subsection{Random Verification with Exponential Distribution}

In this section, we consider the case that donor join the
verification by following exponential distribution. The people who
donate more money have higher probability to do the verification
than the people who donate less money.

\begin{theorem}
Assume integer $k\ge 2$ and there are $n$ donors. Each verify takes
$O(kh$ steps, and reports error if it is not an
$(1-\epsilon)$-approximation in the report with probability at least
$1-\delta$, where $\delta =O(e^{-\lambda \epsilon M})$ and $h$ is
the depth of the tree.
\end{theorem}

\begin{proof}
Let $M=D(R)$ where $R$ is the root of the donation tree. Assume that
there is at least $\epsilon M$ error ($D(R)-V(R)\ge \epsilon D(R)$).

Let $m_1,m_2,\cdots, m_t$ be nodes  with error paths to root of the
tree, and have $\sum_{i=1}^tm_i\ge \epsilon M$ by
Lemma~\ref{count-lemma}.

If one of $m_1,m_2,\cdots, m_t$  checks its path to the root, then
an error (or cheating ) can be detected. Therefore, this problem
becomes to compute the probability that none of $m_1,m_2,\cdots,
m_t$ does his verification.

The probability that none of them checks is at most $e^{-\lambda
m_1}\cdot e^{-\lambda m_2}\cdots e^{-\lambda m_t}\le e^{-\lambda
\epsilon M}$.

\end{proof}

\subsection{An Implementation with B-Tree}

A donation tree can be implemented with a B-tree that supports
$O(\log n)$ time for searching, insertion, and deletion. When an new
leave is inserted, we can update all $V(N)$ for node $N$ affected in
$O(\log n)$ time. Similarly, When an new leave is deleted, we can
update all $V(N)$ for node $N$ affected in $O(\log n)$ time.

\subsection{Uniform Random Verification}

In this section, we consider the case that donor join the
verification by following uniform distribution.

\begin{theorem}\label{uniform-theorem}
Assume that each person donates the money in the range $[1,a]$,
 each donor participates in the
verification with probability at least $\delta$. Then it takes
$O(\log_k n)$ steps, and reports error with probability at least
$1-(1-\delta)^{\ceiling{eM\over a}}$.
\end{theorem}

\begin{proof}
Assume that $M$ is the total amount of money donated by all the
people. If it is not an $(1-\epsilon)$-approximation, then there are
at least $\epsilon M\over a$ error paths corresponding to at least
$k=\ceiling{\epsilon M\over a}$ donors. With probability at most
$(1-\delta)^k=(1-\delta)^{\ceiling{\epsilon M\over a}}$, none of
them will attend the verification. Therefore, with probability at
least $1-(1-\delta)^{\ceiling{eM\over a}}$, the error of the report
will be detected.
\end{proof}

\subsection{Multiple Verification Regions}

In this section, we show the verification in several region. If each
person donates amount in the range $[a_0, a]$. The interval is
partitioned into $[a_0, a_1), [a_1, a_2),\cdots, [a_{k-1}, a_k]$. We
assume people different region have different probability to
participate the verification.

\begin{theorem}
Assume that $[a_0, a_1), [a_1, a_2),\cdots, [a_{k-1}, a_k]$ form a
partition for $[a_0, a]$ with $a_i+1\le a_i(1+\delta)$ for
$i=0,1,2,\cdots, k-1$. Let $I_j=[a_{j}, a_{j+1})$ if $j<k$, and
$I_k=[a_{k-1}, a_k]$. Let $p_j$ be the probability that a person
with donation range in $I_j$ verifies. Then there is a verification
protocol such that with probability at most
$\sum_{j=0}^{k-1}(1-p_j)^{\epsilon M_j/(1+\delta)}$ to fail to check
$1-\epsilon$ approximation, where $M_j$ is the total amount of
donation with each donation in $I_j$. Furthermore, the verification
time is $O(\log n+k)$.
\end{theorem}

\begin{proof}
Use one verification tree $T_j$ for each $I_j$. Form a tree $T$ by
linking $T_1,\cdots, T_{k-1}$ as children. It follows from
Theorem~\ref{uniform-theorem}.
\end{proof}

\vskip 30pt

\begin{figure}
{ \vskip 100pt

  {\begin{picture}(10.0,10.0)



    \put(185.0, 100.0){\begin{picture}(0.0,0.0)
                           \line(1,0){60.0}
                         \end{picture}
                        }

    \put(185.0, 130.0){\begin{picture}(0.0,0.0)
                           \line(1,0){60.0}
                         \end{picture}
                        }

    \put(185.0, 100.0){\begin{picture}(0.0,0.0)
                           \line(0,1){30.0}
                         \end{picture}
                        }

    \put(245.0, 130.0){\begin{picture}(0.0,0.0)
                           \line(0,-1){30.0}
                         \end{picture}
                        }

    \put(210.0, 120.0){a}
    \put(195.0, 105.0){\small V(a)=96}


    \put(95.0, 50.0){\begin{picture}(0.0,0.0)
                           \line(1,0){60.0}
                         \end{picture}
                        }

    \put(95.0, 80.0){\begin{picture}(0.0,0.0)
                           \line(1,0){60.0}
                         \end{picture}
                        }

    \put(95.0, 50.0){\begin{picture}(0.0,0.0)
                           \line(0,1){30.0}
                         \end{picture}
                        }

    \put(155.0, 80.0){\begin{picture}(0.0,0.0)
                           \line(0,-1){30.0}
                         \end{picture}
                        }

    \put(120.0, 70.0){b}
    \put(105.0, 55.0){\small V(b)=6}

     \put(150.0, 80.0){\begin{picture}(0.0,0.0)
                          \vector(2,1){40.0}
                        \end{picture}
                       }


    \put(275.0, 50.0){\begin{picture}(0.0,0.0)
                           \line(1,0){60.0}
                         \end{picture}
                        }

    \put(275.0, 80.0){\begin{picture}(0.0,0.0)
                           \line(1,0){60.0}
                         \end{picture}
                        }

    \put(275.0, 50.0){\begin{picture}(0.0,0.0)
                           \line(0,1){30.0}
                         \end{picture}
                        }

    \put(335.0, 80.0){\begin{picture}(0.0,0.0)
                           \line(0,-1){30.0}
                         \end{picture}
                        }

    \put(300.0, 70.0){c}
    \put(285.0, 55.0){\small V(c)=90}

    \put(278.0, 80.0){\begin{picture}(0.0,0.0)
                          \vector(-2,1){40.0}
                        \end{picture}
                       }


    \put(35.0, 00.0){\begin{picture}(0.0,0.0)
                           \line(1,0){60.0}
                         \end{picture}
                        }

    \put(35.0, 30.0){\begin{picture}(0.0,0.0)
                           \line(1,0){60.0}
                         \end{picture}
                        }

    \put(35.0, 0.0){\begin{picture}(0.0,0.0)
                           \line(0,1){30.0}
                         \end{picture}
                        }

    \put(95.0, 30.0){\begin{picture}(0.0,0.0)
                           \line(0,-1){30.0}
                         \end{picture}
                        }

    \put(60.0, 20.0){d}
    \put(45.0, 5.0){\small V(d)=1}

     \put(70.0, 30.0){\begin{picture}(0.0,0.0)
                          \vector(2,1){40.0}
                        \end{picture}
                       }


    \put(135.0, 00.0){\begin{picture}(0.0,0.0)
                           \line(1,0){60.0}
                         \end{picture}
                        }

    \put(135.0, 30.0){\begin{picture}(0.0,0.0)
                           \line(1,0){60.0}
                         \end{picture}
                        }

    \put(135.0, 00.0){\begin{picture}(0.0,0.0)
                           \line(0,1){30.0}
                         \end{picture}
                        }

    \put(195.0, 30.0){\begin{picture}(0.0,0.0)
                           \line(0,-1){30.0}
                         \end{picture}
                        }

    \put(160.0, 20.0){e}
    \put(145.0, 5.0){\small V(e)=5}

     \put(175.0, 30.0){\begin{picture}(0.0,0.0)
                          \vector(-2,1){40.0}
                        \end{picture}
                       }


    \put(235.0, 00.0){\begin{picture}(0.0,0.0)
                           \line(1,0){60.0}
                         \end{picture}
                        }

    \put(235.0, 30.0){\begin{picture}(0.0,0.0)
                           \line(1,0){60.0}
                         \end{picture}
                        }

    \put(235.0, 0.0){\begin{picture}(0.0,0.0)
                           \line(0,1){30.0}
                         \end{picture}
                        }

    \put(295.0, 30.0){\begin{picture}(0.0,0.0)
                           \line(0,-1){30.0}
                         \end{picture}
                        }

    \put(260.0, 20.0){f}
    \put(245.0, 5.0){\small V(f)=10}

     \put(255.0, 30.0){\begin{picture}(0.0,0.0)
                          \vector(2,1){40.0}
                        \end{picture}
                       }


    \put(335.0, 00.0){\begin{picture}(0.0,0.0)
                           \line(1,0){60.0}
                         \end{picture}
                        }

    \put(335.0, 30.0){\begin{picture}(0.0,0.0)
                           \line(1,0){60.0}
                         \end{picture}
                        }

    \put(335.0, 00.0){\begin{picture}(0.0,0.0)
                           \line(0,1){30.0}
                         \end{picture}
                        }

    \put(395.0, 30.0){\begin{picture}(0.0,0.0)
                           \line(0,-1){30.0}
                         \end{picture}
                        }

    \put(360.0, 20.0){g}
    \put(345.0, 5.0){\small V(g)=84}

     \put(355.0, 30.0){\begin{picture}(0.0,0.0)
                          \vector(-2,1){40.0}
                        \end{picture}
                       }

   \end{picture}
  }\caption{Donation Tree with Cheating. Both nodes f and g can find
the cheating problem at their paths to the root. For example, at
node c, $V(c)<V(f)+V(g)$.}\label{fig3} }
\end{figure}
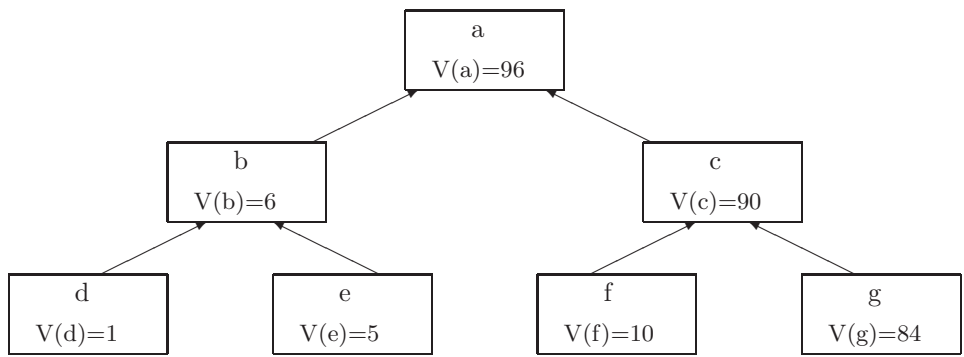

\section{Impossibilities of Verification}

In this section, we show that it is impossible to use uniform
probability to do donation verification. We also prove that it is
impossible to do verification if negative items are allowed.

\begin{theorem}
There is no randomized algorithm fail to detect the cheating from
collector with probability at most $\delta$ if every donor checks
his donation with probability at most $\delta$.
\end{theorem}

\begin{proof}
Let $k=9$. Imagine the collector receives $M$ amount money with
${M\over k}$ from one donor A. He releases a document that includes
all the money from the others except A. If A does not verify it, it
should be all correct without any error. Therefore, with probability
at most $\delta$, the verification fails.
\end{proof}

\begin{theorem}
There is no randomized algorithm fail to detect the cheating with
negative donation allowed from collector with probability at most
$\delta$ if one donor checks his donation with probability at most
$\delta$.
\end{theorem}

\begin{proof}
Let the sum of $n-2$ donors $m_1,\cdots, m_{n-2}$ be equal to $M$.
Let the donor $n-1$ contributes $1$ or $0$, and donor $n$
contributes $-M$. Consider the first case that donor $n-1$
contributes $1$. The total is equal to $1$.

Consider the second case that donor $n-1$ contributes $0$. The total
is equal to $0$. If that donor $n-1$ takes probability at most
$\delta$ to do verification, then we have probability at most
$\delta$ to make the difference of the two cases.

\end{proof}

\section{Conclusions}

In this paper, we develop a protocol for the donation verification
under some probabilistic assumption. It only expects the donors
follow certain probabilistic distribution to attend verification,
and takes $O(\log n)$ steps for each donor.


\begin{thebibliography}{999}
\bibitem{Wojciechowski09}
Adam Wociechowski, Models of Charity Donations and Project Funding
in Social Networks, Lecture Notes in Computer Science 5872, pp.
454-463, 2009.

\bibitem{CooterBroughman05}
Robert Cooter, and Brian J Broughman, Charity, Publicity, and the
Donation Registry, The Economists' Voice. Volume 2, Issue 3, pp.
1553-3832,
2005.


\bibitem{RanganathanHenley07}
Sampath Kumar Ranganathan and Walter H. Henley, Determiniants of
charitable donation intentions: a structural equation model,
International Journal of Nonprofit and Voluntary Sector Marketing,
DOI: 10.1002/nvsm.297,
2007.

\bibitem{OlsenPracejusBrown03}
G. Dougla Olsen, John W. Pracejus, and Norman R. Brown, When profit
equals price: consumer confusion about donation amounts in
cause-related marketing. Journal of Public Policy \& Marketing: Vol.
22, No. 2, pp. 170-180,
2003.

\bibitem{Bennett09}
Rogers Bennett, Factors influencing donation switching behaviour
among charity supporters: an empirical investigation, Journal of
Customer Behaviour, Volume 8, Number 4, pp. 329-345, 2009.

\end{thebibliography}

\end{document}